\documentclass[11pt]{article}   
\newtheorem{mydef}{\bf Definition}
\newtheorem{mythm}{\bf Theorem}
\newtheorem{myprob}{\bf Problem}

\newtheorem{mypro}{\bf Proposition}
\newtheorem{myexm}{\bf Example}
\newtheorem{remark}{\bf Remark}
\newtheorem{proof}{\bf Proof}

\usepackage{subfig}
\usepackage{graphicx}
\graphicspath{{./Pics/}}
\DeclareGraphicsExtensions{.pdf}
\usepackage{amsmath}
\usepackage{amsfonts}
\usepackage{arydshln}
\usepackage{algpseudocode}
\usepackage{verbatim}
\usepackage{enumerate}
\usepackage{rotating}
\usepackage{color}
\usepackage{caption} 
\usepackage{mathrsfs}
\usepackage{amssymb}
\usepackage[linesnumbered,ruled,vlined]{algorithm2e}
\usepackage{float}
\usepackage{hyperref}
\hypersetup{hidelinks} 
\usepackage{soul}
\usepackage{booktabs}
\usepackage{multirow}
\usepackage{makecell}
\usepackage{array}
\usepackage{tcolorbox}
\usepackage{fancyhdr}
\usepackage{framed}
\usepackage[para,online,flushleft]{threeparttable}
\usepackage{flushend}
\usepackage{cite}
\usepackage[margin=1in]{geometry}

\title{\LARGE \bf Efficient STL Control Synthesis under \\ Asynchronous Temporal Robustness Constraints}

\author{Xinyi Yu, Xiang Yin, and Lars Lindemann
	\thanks{This work was supported by  the National Natural Science Foundation of China (62061136004, 62173226, 61833012). 
Xinyi Yu and Xiang Yin are with Department of Automation and Key Laboratory of System Control and Information Processing, Shanghai Jiao Tong University, Shanghai 200240, China.
	e-mail: {\tt\small $\{$yuxinyi-12, yinxiang$\}$@sjtu.edu.cn} 
  Lars Lindemann is with Thomas Lord Department of Computer Science, University of Southern California, Los Angeles, CA 90089, USA.
	e-mail: \tt\small llindema@usc.edu}
}
\date{}

\begin{document}
	
\maketitle
\pagestyle{plain}
\begin{abstract}
In time-critical systems, such as  air traffic control systems, it is crucial to design control policies that are robust to timing uncertainty. Recently, the notion of Asynchronous Temporal Robustness (ATR) was proposed to capture the robustness of a system trajectory against individual time shifts in its sub-trajectories. In a multi-robot system, this may correspond to individual robots being delayed or early. Control synthesis under ATR constraints is challenging and has not yet been addressed. In this paper, we propose an efficient control synthesis method under ATR constraints which are defined with respect to simple safety or  complex signal temporal logic specifications.
Given an ATR bound, we compute a sequence of control inputs so that the specification is satisfied by the system as long as each sub-trajectory is shifted not more than the ATR bound. We avoid combinatorially exploring all  shifted sub-trajectories by first identifying redundancy between them. We capture this insight by the notion of instant-shift pair sets, and then propose an optimization program that enforces the specification only over the instant-shift pair sets. We show soundness and completeness of our method and  analyze its computational complexity.  Finally, we present various illustrative case studies. 
\end{abstract}

\section{Introduction}\label{sec:intro}
% \subsection{Background and a Motivating Example}\label{subsec:back}

The analysis and design of distributed cyber-physical systems has found much attention in industry and academia in recent years. These systems are often time-critical and require accurate timing in the satisfaction of system requirements \cite{sha2004real}, e.g., in air traffic control systems and multi-robot systems. Therefore, it is important to ensure robustness of the employed control algorithms against timing uncertainties.

Consider the scenario that is shown in Fig.~\ref{fig:illustration} where
the red agent starts from the lower left corner and is required to arrive in the cyan region (represented by the constraint $\mu_1(x)\geq 0$) between the time interval $6-9$  to perform an action. The blue agent, on the other hand, starts from the lower right corner and needs to track the red agent with an accuracy of $D$ meters between the time interval $6-9$. 
% (If the red trajectory is fixed, then the blue agent should stay in the blue region.)
% two agents starting from the lower left and the lower right corners, respectively, need to first gather at the cyan region (represented by the function $\mu_1(x)\geq 0$) between time period 3 to 8 to achieve some tasks, e.g., exchange something important by hardware or collaborate to finish something, and 
Eventually, the red and blue agents should arrive in the grey top left and the grey top right regions (represented by the constraints $\mu_2(x)\geq 0$ and $\mu_3(x)\geq 0$), respectively, between the time interval $17-20$. Fig.~\ref{fig:illustration} shows two pre-planned paths that satisfy this task. However, during the real-time execution of this path, timing uncertainty may lead to a task violation. For example, an agent may  be delayed due to a system failure or too early due to velocity control errors. In fact, if the blue agent is delayed by two time steps, then the distance  to the red agent may be larger than $D$. In this work, we propose a control synthesis method to compute temporally robust trajectories under high-level tasks described by constraint functions or complex signal temporal logic specifications. 

\begin{figure}[t]
	\centering
	\includegraphics[width = 230pt]{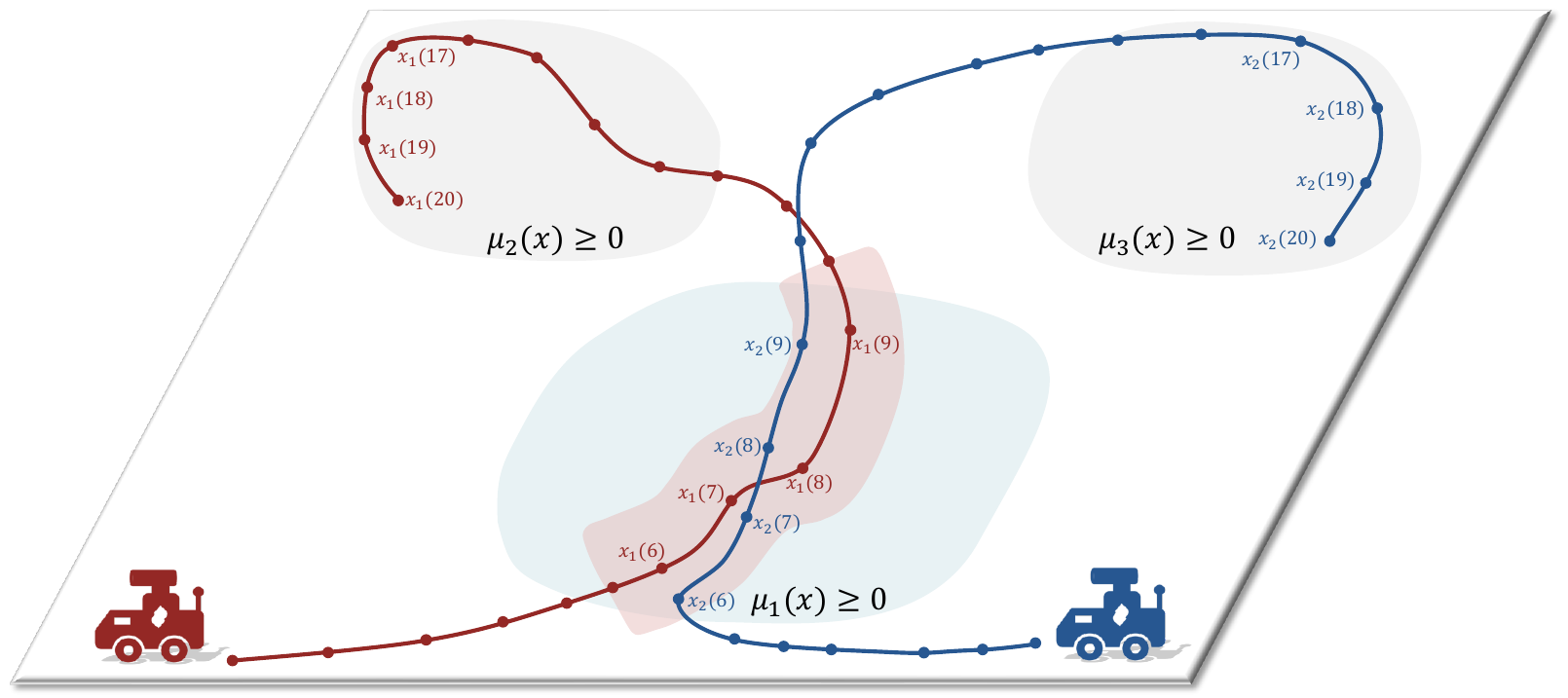}
        \vspace{3pt}
	\caption{Case study illustrating the effect of timing uncertainties, e.g., if the blue agent is delayed and not able to track the red agent during a pre-specified tracking interval.}
	\label{fig:illustration}
\end{figure}

\subsection{Related Work}
Control synthesis problems of time-critical systems have been studied from various perspectives. In one direction, the minimum time control problem has been well studied in the literature, see e.g., \cite{dubowsky1989planning, faroni2022safety, van2011model, liu2022time}. While minimum time control is important in many cases, it does not account for timing uncertainties.  In another direction, resilient control synthesis studies how to quickly recover from system failures 
\cite{chen2022stl, aksaray2021resilient}. The authors in these works consider not only simple navigation tasks, but also complex Signal Temporal Logic (STL) specifications \cite{maler2004monitoring} as we do in our work. In yet another direction, the works in \cite{buyukkocak2022temporal, penedo2020language, vasile2017time} consider temporal relaxations to find minimally violating solutions. However, while all of these works consider temporal aspects in the control design, none of these works directly consider time shifts in the agent trajectory as e.g. needed in our motivating example.

For this reason, different authors have investigated how to preserve task satisfaction under such time shifts, e.g., as in 
\cite{desai2017drona}. Robust control synthesis under time shifts in agent trajectories is studied for Linear Temporal Logic (LTL) specifications in \cite{ulusoy2012robust,sahin2019multirobot, sahin2017synchronous}. These works design control laws which can also account for time scaling effects, e.g., when an agent pauses or speeds up, and not only time shifts.  In our work, we focus on efficient control synthesis under STL specifications.

Temporally robust synthesis under STL specifications is considered in \cite{lin2020optimization} by jointly maximizing spatial and temporal robustness. Specifically, temporal robustness for STL is proposed in \cite{donze2010robust} and control synthesis under these constraints is done in our previous works \cite{rodionova2021time, rodionova2022combined, rodionova2022temporal}. However, the temporal robustness notion in these works considers \emph{time shifts in the atomic elements} (called predicates) of STL which are not always interpretable in terms of \emph{time shifts in the trajectory}.
For example, in the motivating example introduced in Fig. \ref{fig:illustration}, the \emph{predicate} describing the relative position of the red and blue robot (staying $D$-close) contains two \emph{trajectories}. 
The works \cite{rodionova2021time, rodionova2022combined, rodionova2022temporal} consider the shift of this \emph{predicate} as opposed to the shift of the two trajectories, i.e., it cannot deal with shifts in the \emph{trajectory}. 
% , e.g., the motivating example introduced in Fig. \ref{fig:illustration} where the relative position of the red and blue robot is considered (staying $D$-close) can not be dealt with methods in \cite{rodionova2021time, rodionova2022combined, rodionova2022temporal}, since such a \emph{predicate} contains two \emph{trajectories}.
We are inspired by the notion of asynchronous temporal robustness from \cite{lindemann2022temporal} that considers asynchronous time shifts in the agent trajectories, which is in spirit similar to \cite{sahin2019multirobot, sahin2017synchronous}.

\subsection{Contributions}
We propose an efficient method for control synthesis under asynchronous temporal robustness constraints. This notion was recently proposed in \cite{lindemann2022temporal} for temporal robustness analysis by capturing timing uncertainty via time shifts in the system's sub-trajectories. Our contributions are as follows:
\begin{itemize}
    \item We propose an optimization algorithm that synthesizes temporally robust trajectories $x(k)$ under constraints of the form  $c(x(k))\ge 0$. The encoding is efficient as it avoids combinatorially exploring all time shifts. \vspace{-6pt}
    \item We extend the idea from the above algorithm to complex signal temporal logic specifications, and we synthesize temporally robust trajectories under such specifications. The  algorithm is a Mixed Integer Program (MIP). \vspace{-6pt}
    \item We analyze the computational complexity of both algorithms and show their soundness and completeness. \vspace{-6pt}
    \item We present case studies that showcase the efficacy of our method. We empirically analyze the computational complexity and observe a computational speed-up of more than a factor of ten compared to a naive approach.
\end{itemize}

\subsection{Organization}
 We present background and problem formulation in Section \ref{sec:prob}. We then solve the temporally robust control problem under asynchronous temporal robustness constraints for constraint functions $c$ and signal temporal logic specifications in Sections \ref{sec:sol_c} and \ref{sec:sol_stl}, respectively. Simulations are presented in Section \ref{sec:case}, and Section \ref{sec:con} concludes our work.

\section{Background and Problem Formulation}\label{sec:prob}
\textbf{Notation}: We denote by $\mathfrak{F}(T, S)$ the set of all functions mapping from the domain $T$ into the domain $S$. Let $\mathbb{R}$ and $\mathbb{Z}$ be the set of real numbers and integers, respectively. As we consider a discrete-time setting, we use $[a,b]$ as a shorthand notation for the discrete-time interval $[a,b]\cap \mathbb{Z}$.

\subsection{System Model}

We consider a team of $N$ agents where each agent $i\in\{1,\hdots,N\}$ is described by the discrete-time dynamics 
\begin{align}
	x_i(k+1) = f_i(x_i(k), u_i(k)),\nonumber
\end{align}
where $x_i: \mathbb{Z} \to \mathcal{X}_i \subseteq \mathbb{R}^{n_i}$ is the $n_i$-dimensional state and $u_i: \mathbb{Z} \to \mathcal{U}_i \subset \mathbb{R}^{m_i}$ is the $m_i$-dimensional control input. We can compactly write the system by stacking the dynamics as 
\begin{equation}\label{eq:system}
	x(k+1) = f(x(k), u(k)),
\end{equation}
where 
$x := [x_1, \dots, x_N]: \mathbb{Z} \to \mathcal{X} \subseteq \mathbb{R}^{n}$ and $u := [u_1, \dots, u_N]: \mathbb{Z} \to \mathcal{U} \subset \mathbb{R}^{m}$ are the stacked state and control inputs with corresponding domains $\mathcal{X}:=\mathcal{X}_1\times \hdots\times \mathcal{X}_N$ and $\mathcal{U}:=\mathcal{U}_1\times \hdots\times \mathcal{U}_N$ and dimensions $n:=\Sigma_{i=1}^N n_i$ and $m:=\Sigma_{i=1}^N m_i$. Consequently, we have  $f(x(k), u(k)) := [f_1(x_1(k), u_1(k)), \dots, f_N(x_N(k), u_N(k))]$.

We focus on the multi-agent system in equation \eqref{eq:system} as it makes time shifts, which are formally introduced later, interpretable.
For instance, the agent $i$ is delayed by $\kappa_i$ time units if the state at time $k$ is $x_i(k-\kappa_i)$, $\kappa_i > 0$.  
% For instance, the agent $i$ is delayed by $\kappa_i<0$ time units if $x_i(k+\kappa_i)$. 
If there were dynamical couplings, e.g., if $f_i$  also depends on $x_j$ for $j\neq i$, this interpretation would not be possible. On the other hand, note that an agent may be described by a state that can be decoupled further, e.g., the two-dimensional motion of an omnidirectional robot that can be controlled independently.
One could, in principle, split these agent dynamics further into two sub-systems and consider separate time shifts.
% The consideration of this system form is motivated by the observation that we may have $N$ agents with their own dynamic systems and they can be different. 
%If one agent has several decoupled states, e.g., positions in X and Y directions, we can extend the framework proposed in this paper by setting decoupled parts as individuals. For readability, we consider the simple case which is the stack of $N$ agents model here.

\subsection{Task specification}
We consider two forms of task specifications for the system in \eqref{eq:system}. We first consider tasks expressed by the constraint function $c:\mathbb{R}^n \to \mathbb{R}$, and then consider more complex tasks in the form of Signal Temporal Logic (STL).

\textbf{Constraint functions.} 
% One  way of defining tasks is by means of a constraint $c(x(k), k) \geq 0$ which is imposed for all times $k \!\in\! \mathbb{Z}$. 
One way of defining tasks is by means of constraints $c_j(x(k)) \geq 0$ which are imposed for all times $k$ specified by a corresponding time set $\mathbb{T}_j$ for $j \!\in\! \{1, \dots, n_c\}$ where $n_c$ is the number of constraint functions. We denote by $c$ the set of all tasks consisting of the function $c_j(x(k)), j \!\in\! \{1, \dots, n_c\}$, and we note that $c$ is satisfied if
\[
     c_j(x(k)) \geq 0, \;\; \forall j \!\in\! \{1, \dots, n_c\}, \forall k \!\in\! \mathbb{T}_j.    
\]
For later convenience, we now define the satisfaction function $\beta^c: \mathfrak{F}(\mathbb{Z}, \mathbb{R}^n) \to \{-1, 1\}$ that indicates if the state sequence $x$ satisfies the task $c$ as
\begin{align}
	\beta^c(x) := \left\{
		\begin{array}{cl}
			1 & \text{if} \ \inf_{k\in \mathbb{T}_j} c_j(x(k)) \geq 0, \forall j \!\in\! \{1, \dots, n_c\}\\
			-1 & \text{otherwise}.
		\end{array}
	\right. \nonumber
\end{align}
% \begin{align}
% 	\beta^c(x) := \left\{
% 		\begin{array}{cl}
% 			1 & \text{if} \ \inf_{k\in \mathbb{Z}} \ c(x(k), k) \geq 0 \\
% 			-1 & \text{otherwise}.
% 		\end{array}
% 	\right. \nonumber
% \end{align}
Note that $\beta^c(x) = 1$ represents task satisfaction, while $\beta^c(x) = -1$ indicates task violation. 
% Finally, we denote by $\mathbb{T}^c$ the set of effective time instants of $c$ as $\mathbb{T}^c := \{k\in\mathbb{Z} \mid \nexists d\in\mathbb{R} \text{ s.t. } \forall x\in \mathcal{X},  c(x,k) = d\}$, i.e., the set of times $k$ where the satisfaction of $c(x,k)$ depends on $x$.
% In this paper, we assume that the set $\mathbb{T}^c$ is bounded.
In this paper, we assume that all the time sets $\mathbb{T}_j, j \!\in\! \{1, \dots, n_c\}$ are bounded.

% \begin{myexm}\label{ex:1}
% 	Consider the setting from the motivating example in Section \ref{sec:intro} where the task can be expressed as
%     \begin{align}
%         c(x,k) := \left\{
%             \begin{array}{cl}
%                 h_1(x(k)), & \text{for all } \ k \in [6,9] \\
%                 h_2(x(k)), & \text{for all } \ k \in [17,20] \\
%                 1. & \text{otherwise},
%             \end{array}
%         \right. \nonumber
%     \end{align}
%     with $h_1$ and $h_2$ being defined as 
%     \begin{align}
%         h_1(x(k)) & := \min(\mu_1(x_1(k)), D - ||x_1(k) - x_2(k)||_{2}) \nonumber \\
%         h_2(x(k)) & := \min(\mu_2(x_1(k)), \mu_3(x_2(k))),  \nonumber
%     \end{align}
%     where $D$ is the required maximum distance between two agents.
%     The red and blue agents' states are $x_1$ and $x_2$, respectively, and  the effective time instants set is $\mathbb{T}^c=[4,8] \cup [12,15]$.
% \end{myexm}

\begin{myexm}\label{ex:1}
	Consider the setting from the motivating example in Section \ref{sec:intro} where the task $c$ with $n_c=2$ can be expressed by $c_1(x(k))$ and $c_2(x(k))$ as follows,
    \begin{align}
        c_1(x(k)) & := \min(\mu_1(x_1(k)), D - ||x_1(k) - x_2(k)||_{2}), \nonumber \\
        c_2(x(k)) & := \min(\mu_2(x_1(k)), \mu_3(x_2(k))), \nonumber
    \end{align}
    with $\mathbb{T}_1:= [6,9]$ and $\mathbb{T}_2:= [17,20]$
    where $D$ is the required maximum distance between two agents.
    The red and blue agents' states are $x_1$ and $x_2$, respectively.
    % , and  the effective time instants set is $\mathbb{T}^c=[4,8] \cup [12,15]$.
\end{myexm}

\textbf{Signal Temporal Logic.} Signal temporal logic was introduced in \cite{maler2004monitoring}, and we here only give a brief introduction.
The syntax of STL formulae is as follows 
\[
\phi ::=   \top\mid \pi^\mu \mid \neg \phi \mid \phi_1 \wedge \phi_2 \mid \phi_1 \textbf{U}_{[a,b]} \phi_2,
\]
where $\top$ is the \textsf{true} symbol and $\pi^\mu$ is a predicate whose truth value is determined by the sign of an underlying predicate function $\mu:\mathbb{R}^n \to \mathbb{R}$, i.e., $\pi^\mu=\top$ if and only if $\mu(x(k)) \geq 0$.
The symbols $\neg$ and $\wedge$ denote the standard Boolean operators ``negation" and ``conjunction", respectively, and $\textbf{U}_{[a,b]}$ is the temporal operator ``\emph{until}" with $a,b\in \mathbb{Z}$. We consider bounded STL formulae where the time interval $[a,b]$ is bounded. These operators can be used to define ``disjunction"  by $\phi_1 \vee \phi_2:=\neg(\neg \phi_1 \wedge \neg \phi_2)$, ``implication" by $\phi_1 \to \phi_2:= \neg \phi_1 \vee   \phi_2$, ``eventually" by $\mathbf{F}_{[a,b]} \phi:= \top \mathbf{U}_{[a,b]} \phi$ and ``always" by $\mathbf{G}_{[a,b]} \phi:=\neg \mathbf{F}_{[a,b]} \neg \phi$. 
The semantics of an STL formula $\phi$  can, similar to the constraint function $c$, be captured by a satisfaction function $\beta^\phi: \mathfrak{F}(\mathbb{Z}, \mathbb{R}^n)  \to \{-1,1\}$ whose formal definition is omitted due to space limitation and details can be found in \cite{maler2004monitoring}. For a bounded STL formula $\phi$, the function $\beta^\phi(x)$ is completely determined by a finite state sequence $x$ of length $T_\phi$, which is the \emph{formula horizon} as defined in \cite{dokhanchi2014line, sadraddini2015robust}. 

\begin{myexm}
	Consider again the workspace from Section \ref{sec:intro} as well as the following STL formula
    % \begin{align}
    %     \phi \!:=\! \mathbf{G}_{[6,9]} \pi^{\mu_2^r} \wedge \mathbf{G}_{[6,9]} \pi^{\mu_3^b}  \wedge  \mathbf{F}_{[15, 17]} \mathbf{G}_{[0, 3]} (\pi^{\mu_1^r} \wedge \pi^{\mu_4} )
    %     , \nonumber
    % \end{align}
    \begin{align}
        \phi \!:=\! \mathbf{G}_{[6,9]} \pi_1^{\mu_2} \wedge \mathbf{G}_{[6,9]} \pi_2^{\mu_3}  \wedge  \mathbf{F}_{[15, 17]} \mathbf{G}_{[0, 3]} (\pi_1^{\mu_1} \wedge \pi^{\mu_4} )
        , \nonumber
    \end{align}
    where the subscript $i\!\in\! \{1,2\}$ within $\pi_i^\mu$ means that the predicate $\pi^\mu$ is evaluated over the state $x_i(k)$ of agent $i$, and   
    $\mu_4 := D - ||x_1(k) - x_2(k)||_{2}$. In words, the formula $\phi$ means ``red and blue agents should stay in their respective grey areas within the time interval $[6, 9]$, and after that starting from a certain time instant between the time interval $[15, 17]$, the red agent should stay in the cyan region for at least 4 time instants while the distance between the red and the blue agent is at most $D$ meters''. Note that $T_\phi=20$ and that $\phi$ is more complex than the task in Example \ref{ex:1}.
\end{myexm}
\vspace{6pt}

\subsection{Asynchronous Temporal Robustness (ATR)}
Asynchronous temporal robustness was introduced in \cite{lindemann2022temporal} to capture robustness of $x$ against time shifts with respect to task satisfaction.
To capture time shifts in $x$, we introduce the definition of the time shifted state. 
% \begin{mydef}[State Shift]\label{def:shift}
Given the state $x$, we denote by $x^{\bar{\kappa}}$ the time shifted stacked state with time shift $\bar{\kappa} := (\kappa_1, \dots, \kappa_N)\in\mathbb{R}^N$. Formally, we let
\begin{align}\label{eq:time_shift}
	x^{\bar{\kappa}}(k) := [x_1(k+\kappa_1), \dots, x_N(k+\kappa_N)]. 
\end{align}
Intuitively, the time shifted state $x^{\bar{\kappa}}$ is the stacked state with each state $x_i$ shifted by $\kappa_i$ where $\kappa_i<0$ will indicate a delay and $\kappa_i>0$ will indicate being early. This definition is useful in time-critical systems, e.g., where a robot is delayed by $\kappa_i$ to perform an action.  
We remark that the decoupled dynamics in \eqref{eq:system} enable us to 
define the shifted state in this way. We will consider control synthesis for a finite time horizon of $T$ and hence modify the time shifted state $x_i(k+\kappa_i)$ to
\begin{align}\label{eq:xi_shift}
	\!\!\!x_i(k\!+\!\kappa_i) \!:=\! \left\{
		\begin{array}{ll}
            \!\! x_i(k\!+\!\kappa_i) & \!\!\!\! \text{if} \ k\!+\!\kappa_i\!\in\![0,T]\\
		\!\! x_i(0) & \!\!\!\! \text{if} \ k\!+\!\kappa_i < 0 \\
            \!\! f_i(x_i(k\!+\!\kappa_i\!-\!1), \mathbf{0}^{m_i}) & \!\!\!\! \text{if} \ k\!+\!\kappa_i>T.
		\end{array}
	\right. 
\end{align}
	% and $f_i: \mathbb{R}^{n_i} \times \mathbb{R}^{m_i} \to \mathbb{R}^{n_i}$ with the number of dimension $n_i$ in $i$th group represents the transition relation of $i$th group state which has nothing to do with states in other dimensions under $N$ group state assumption. 
% \end{mydef}
% \vspace{3pt}
Particularly, we set the time shifted state to the initial state when $k+\kappa_i<0$ and we forward propagate the state under the zero control input when $k+\kappa_i>T$.  It is worth mentioning that the last condition can be changed to  other reasonable forms, e.g., $x_i(k+\kappa_i)=x_i(k+\kappa_i-1)$.

\emph{Asynchronous Temporal Robustness} \cite{lindemann2022temporal} is used to capture the maximum permissible time shift in the state sequence while the task is still satisfied.
The asynchronous temporal robustness (ATR) $\theta: \mathfrak{F}(\mathbb{Z}, \mathbb{R}^n) \to \mathbb{Z}$ is defined as follows
\begin{align}\label{eq:theta}
		\theta(x) \!:=\! \beta(x) \cdot \max
   \left\{\tau \!\in\! \mathbb{Z} \,\middle\vert\, 
   \begin{array}{cc}
         \!\!\forall  \kappa_1, \dots, \kappa_N \!\in\!   [-\tau, \tau],\\
	  \beta(x^{\bar{\kappa}}) = \beta(x)
  \end{array}
  \right\}, 
\end{align}
where $\beta\in\{\beta^c,\beta^\phi\}$ depends on the task at hand. Consequently, if a state sequence $x$ has an ATR of $\theta':=\theta(x)$ and satisfies the task, i.e., $\beta(x)=1$, then every time shifted signal satisfies the task, i.e., $\beta(x^{\bar{\kappa}})=1$, if $\bar{\kappa}\in[-\theta', \theta']^N$. Note that the interval $[-\theta', \theta']$ is symmetrical, while we, in this paper,  instead consider a more general and asymmetric bound $[\theta_1, \theta_2]$ where we consider $\theta_1$ and $\theta_2$ to be given robustness parameters. This, on the other hand, means that we do not seek the maximum permissible time shift as in equation \eqref{eq:theta}. In the remainder, we assume that $\theta_1\le 0$ (allowing delays) and $\theta_2\ge 0$ (allowing being early).

\subsection{Temporally Robust Control under ATR Constraints}
Given an initial state $x(0)$ of the system in \eqref{eq:system} and an input sequence $u(0\!:\!T-1)$, consider a cost function $J: \mathcal{X} \times \mathcal{U}^{T} \to \mathbb{R}$. Given the robustness bound $[\theta_1,\theta_2]$, the goal of this work is to compute a control sequence $u(0\!:\!T-1)$ which is robust to asynchronous time shifts $\bar{\kappa}\in[\theta_1,\theta_2]^N$, while minimizing for the cost function $J$. Particularly, we consider two problem formulations, one for constraint functions $c$ and one for STL specifications $\phi$.

\begin{myprob}[ATR Problem for Constraint Functions]\label{prob}\upshape
    Given a  system of the form \eqref{eq:system} with initial state $x(0)
    % \!\in\! \mathcal{X}
    $, a cost function $J$, a  task $c$, a time horizon $T \!:=\! \max_{j\in \{1, \dots, n_c\}} \max \mathbb{T}_j$, and an ATR bound $[\theta_1, \theta_2]$ with $\theta_1 \leq 0 \leq \theta_2$, compute the optimal control sequence $u^*(0\!:\!T-1)$ which minimizes the cost function $J$ and is robust to time shifts $[\theta_1, \theta_2]$, i.e., 
	\begin{subequations}
		\begin{align}
			& \underset{u^*(0:T-1)}{\text{argmin}} & & J(x(0), u(0\!:\!T-1)) \nonumber \\
			& \text{subject to} & & \nonumber \\
			& &&
			\!\!\!\!\!\!\!\!\!\!\!\!\!\!\!\!\!\!\!\!\!\!\!
			u(0),u(1),\dots, u(T-1) \in \mathcal{U}, \nonumber\\
			& &&
			\!\!\!\!\!\!\!\!\!\!\!\!\!\!\!\!\!\!\!\!\!\!\!
			x(k+1) = f(x(k), u(k)), k=0, \dots, T-1, \nonumber \\
			& &&
			\!\!\!\!\!\!\!\!\!\!\!\!\!\!\!\!\!\!\!\!\!\!\!
		 \beta^c(x^{\bar{\kappa}}) = 1, \; \forall \kappa_1, \dots, \kappa_N \!\in\!   [\theta_1, \theta_2]. \nonumber
		\end{align} 
	\end{subequations} 
\end{myprob}
\vspace{6pt}

Note that the last constraint ensures i) nominal task satisfaction, i.e., $\beta^c(x)=1$, as we assume $\theta_1 \leq 0 \leq \theta_2$, and ii) satisfaction of the ATR bound $[\theta_1,\theta_2]$.

\begin{myprob}[ATR Problem for STL Tasks]\label{prob:stl}\upshape
    Given a  system of the form \eqref{eq:system} with initial state $x(0) 
    % \!\in\! \mathcal{X}
    $, a cost function $J$, a bounded STL specification $\phi$, a time horizon $T:=T_\phi$, and an ATR bound $[\theta_1, \theta_2]$ with $\theta_1 \leq 0 \leq \theta_2$, compute the optimal control sequence $u^*(0\!:\!T-1)$ which minimizes the cost function and is robust to time shifts $[\theta_1, \theta_2]$, i.e., 
	\begin{subequations}
		\begin{align}
			& \underset{u^*(0:T-1)}{\text{argmin}} & & J(x(0), u(0\!:\!T-1)) \nonumber \\
			& \text{subject to} & & \nonumber \\
			& &&
			\!\!\!\!\!\!\!\!\!\!\!\!\!\!\!\!\!\!\!\!\!\!\!
			u(0),u(1),\dots, u(T-1) \in \mathcal{U}, \nonumber\\
			& &&
			\!\!\!\!\!\!\!\!\!\!\!\!\!\!\!\!\!\!\!\!\!\!\!
			x(k+1) = f(x(k), u(k)), k=0, \dots, T-1, \nonumber \\
			& &&
			\!\!\!\!\!\!\!\!\!\!\!\!\!\!\!\!\!\!\!\!\!\!\!
			\beta^\phi(x^{\bar{\kappa}}) = 1, \; \forall \kappa_1, \dots, \kappa_N \!\in\!   [\theta_1, \theta_2]. \nonumber
		\end{align} 
	\end{subequations} 
\end{myprob}

Now, we note already that the solutions to Problems \ref{prob} and \ref{prob:stl}, which we present in Sections \ref{sec:sol_c} and \ref{sec:sol_stl},  will have different computational complexity. In general, STL control synthesis is  NP-hard and computationally more challenging than dealing with constraint functions $c$.  

\section{Temporally Robust Control for  Constraint Functions}
\label{sec:sol_c}

A straightforward solution to Problem \ref{prob} is to consider all combinations of time shifts $\kappa_1,\hdots,\kappa_N\in[\theta_1, \theta_2]$. Following the definition of the satisfaction function $\beta^c(x)$, we can do so by solving an optimization problem where we introduce the constraint $c_j(x^{\bar{\kappa}}(k))\geq 0$ for each time shift from the set
\[
	\bar{\kappa} \!\in\! \bar{\mathbf{K}} := \{(\kappa_1, \dots, \kappa_N) \mid \forall \kappa_1, \dots, \kappa_N \!\in\! [\theta_1, \theta_2]\}
\]
and for each constraint function $j \!\in\! \{1, \dots, n_c\}$ and for all corresponding effective time instants $k \!\in\! \mathbb{T}_j$, i.e., we enforce
\begin{align}\label{eq:cons}
       c_j(x^{\bar{\kappa}}(k)) \geq 0, \;\; \forall j \!\in\! \{1, \dots, n_c\}, \forall k \!\in\! \mathbb{T}_j, \forall \bar{\kappa} \in \bar{\mathbf{K}}.
\end{align}
However, this method is exhaustive and results in $\Theta^N |\mathbb{T}_j|$ constraints for each  $j \!\in\! \{1, \dots, n_c\}$ where $\Theta := \theta_2-\theta_1+1$  is the length of the interval $[\theta_1,\theta_2]$ and where  $|\mathbb{T}_j|$ denotes the cardinality of the set $\mathbb{T}_j$. The  amount of constraints of this naive approach may be a computational bottleneck for a large number of constraint functions $n_c$, for large time shift intervals $\Theta$, for large number of agents $N$, and for large sets of effective times $\mathbb{T}_j$.  

In this paper, we instead propose a more efficient method by identifying redundancy in \eqref{eq:cons} and by reducing the total number of constraints. While we obtain a more efficient solution for constraint functions $c$, this benefit will be even more significant for robust control synthesis for STL specifications as presented in the next section. Our reduction is based on the following observation that we state as a proposition.

\begin{mypro}\label{prop:1}\upshape
    The state at time $k$ after time shift $\bar{\kappa}$ is equivalent to the state at time $k'$ after time shift $\bar{\kappa}'$, i.e., $x^{\bar{\kappa}}(k) = x^{\bar{\kappa}'}(k')$, if  $k'+\kappa'_i = k+\kappa_i$ for all $i \in \{1, \dots, N\}$.
 \end{mypro}
 
%\begin{mypro}
%    The state at instant $k$ after shift $\bar{\kappa}$ is the same as the state at instant $k'$ after shift $\bar{\kappa}'$ , i.e., $x^{\bar{\kappa}}(k) = x^{\bar{\kappa}'}(k')$, if and only if $\forall i \in \{1, \dots, N\}, k'+\kappa'_i = k+\kappa_i$, where $\kappa_i$ ($\kappa'_i$) is the $i$th item of $\bar{\kappa}$ ($\bar{\kappa}'$).
 %\end{mypro}

From equation \eqref{eq:time_shift} it is easy to see that this result holds. The intuition here is that the time shift $\kappa_i$ at time $k$ evens out with the time shift $\kappa_i'$ at time $k'$, e.g., for $n=N=1$ we have that $x^{(1)}(0) = x^{(-1)}(2)$ are both pointing to $x(1)$.
% \YXY{Note that it will not be much of a problem if $k+\kappa_i\notin [0,T]$ since we defined what the corresponding state is in Equation \eqref{eq:xi_shift}.}
Based on this observation, we know that some of the constraints in \eqref{eq:cons} are repeated, and that we can safely remove these repeated constraints.
% Formally, such a property can be captured by the following proposition and we will utilize it to reduce the number of inequalities later.
To capture the set of time shifted states that contain the same state information, we introduce the notion of the \emph{instant-shift pair set} (``instant'' for time instant $k$ and ``shift'' for time shift $\bar{\kappa}$) as follows.

\begin{mydef}[Instant-Shift Pair Set]\upshape
    Given the time shift set $\bar{\mathbf{K}}$, the instant-shift pair set $I(k,\bar{\kappa})$ is the set of all pairs $(k', \bar{\kappa}')\in \mathbb{T} \times \bar{\mathbf{K}}$ that point to the state $x^{\bar{\kappa}}(k)$, i.e.,
     % \begin{align}\label{eq:pair}
     % I(k,\bar{\kappa})_{\mathbb{T}^c, \bar{\mathbf{K}}} = \{& (k', \bar{\kappa}') \in \mathbb{T}^c \times \bar{\mathbf{K}}\mid  
    	%    \\
    	% & \forall i \in \{1, \dots, N\}, k'+\kappa'_i = k+\kappa_i,  
    	% \}.    \nonumber  
     % \end{align}
     \begin{equation}\label{eq:pair}
      I(k,\bar{\kappa}) \!:= \!\{  (k', \bar{\kappa}') \!\in\! \mathbb{T} \!\times\! \bar{\mathbf{K}}\mid  \!
    	   \forall i \!\in\! \{1, \dots, N\}, k'+\kappa'_i \!=\! k+\kappa_i
    	\},  
     \end{equation} 
     where $\mathbb{T} \in \{\mathbb{T}_j, [0, T_\phi]\}$ is the set of effective times (depending on the choice of the specification).
\end{mydef}

For constraint functions $c_j$, we denote by $\mathbb{P}_j^{s}$ the set of all instant-shift pair sets with respect to $\mathbb{T}_j$, i.e., $\mathbb{P}_j^{s} := \{I(k,\bar{\kappa}) \mid  (k,\bar{\kappa}) \!\in\! \mathbb{T}_j \times \bar{\mathbf{K}} \}$.
Note that each element $d\in\mathbb{P}_j^{s}$ is a set that may consist of more than one pair $(k',\bar{\kappa}')$. For convenience, we define $\mathbb{P}_j := \{d(1) \;|\; d\in \mathbb{P}_j^{s} \}$ which simply selects the first pair $d(1)$ from this set $d$.
% Intuitively, $I(k,\bar{\kappa})$ is an equivalence class for pairs of form $(k,\bar{\kappa})$ such that they essentially correspond to the same time shift. Therefore, we just need to consider any one pair  $(k',\bar{\kappa}')\in I(k,\bar{\kappa})$ for each equivalence class. 
% Formally, we define
% \begin{equation}
% \min\{I(k,\bar{\kappa})\}:=(k',\bar{\kappa}') \nonumber
% \end{equation}
% be such ``smallest" pair in instant-shift pair set such that 
% \[
% \forall (k'',\bar{\kappa}'')\in I(k,\bar{\kappa}): k'\leq k'',
% \]
% which is uniquely defined for each $I(k,\bar{\kappa})$ since there is not two pairs with the same $k$ in $I(k,\bar{\kappa})$ according to Proposition 1.
% Then we define
% \begin{equation}
% \mathbb{P} := \{  \min\{I(k,\bar{\kappa})\}  \in  \mathbb{T} \!\times\! \bar{\mathbf{K}}
% \mid 
% I(k, \bar{\kappa}) \!\in\! \mathbb{T} \!\times\! \bar{\mathbf{K}} \}
% \end{equation}
% to be the set of such smallest pairs in which each pair represents an equivalence class of pairs.
In words, $\mathbb{P}_j$ is the set of representative instant-shift pairs selected from $\mathbb{T}_j \!\times\! \bar{\mathbf{K}}$, but it also includes all the ``state information" as in the set $\mathbb{T}_j \!\times\! \bar{\mathbf{K}}$.
% Additionally, we denote by $\mathbb{P}_s$ the set of all instant-shift pair sets, i.e., $\mathbb{P}_s := \{I(k,\bar{\kappa}) \mid  (k,\bar{\kappa}) \!\in\! \mathbb{T} \times \bar{\mathbf{K}} \}$ which will be used in the next section. 
% \st{Note that each element $d\in\mathbb{P}_s$ is a set that may consists of more than one pair $(k',\bar{\kappa}')$. For convenience, we define $\mathbb{P} := \{d(1) \;|\; d\in \mathbb{P}_s \}$ which simply selects the first pair $d(1)$ from this set $d$. }
% \YXY{Note that each element $I(k,\bar{\kappa}) \in\mathbb{P}_s$ is a set that may consists of more than one instant-shift pair. For convenience, we define $\mathbb{P} := \{(k,\bar{\kappa}) \in I(k,\bar{\kappa}) \in \mathbb{P}_s \;|\; \forall (k',\bar{\kappa}') \in I(k,\bar{\kappa}), \ k \leq k' \}$ which selects the pair with minimum time from each element $I(k,\bar{\kappa})$ in $\mathbb{P}_s$. According to Proposition 1, there is not two pairs with repeated $k$. As a result, only one pair will be chosen to $\mathbb{P}$ from each $I(k,\bar{\kappa})$ in $\mathbb{P}_s$.}

Now, instead of enforcing the constraints $c_j$ for all $j \!\in\! \{1, \dots, n_c\}$, $k \!\in\! \mathbb{T}_j$, and $\bar{\kappa}\in\bar{\mathbf{K}}$ as in \eqref{eq:cons}, we can use the representative instant-shift pair set $\mathbb{P}_j$ to solve Problem \ref{prob}. Particularly, we find an optimal control sequence that satisfies the ATR bound $[\theta_1,\theta_2]$ as follows:
\begin{subequations}\label{eq:opt}
	\begin{align}
		& \underset{u^*(0:T-1)}{\text{minimize}} & & J(x(0), u(0\!:\!T-1)) \\
		& \text{subject to} & & \nonumber \\
		& &&
		\!\!\!\!\!\!\!\!\!\!\!\!\!\!\!\!\!\!\!\!\!\!\!
		u(0),u(1),\dots, u(T-1) \in \mathcal{U}, \\
		& &&
		\!\!\!\!\!\!\!\!\!\!\!\!\!\!\!\!\!\!\!\!\!\!\!
		x(k+1) = f(x(k), u(k)), k=0, \dots, T-1, \\
		& &&
		\!\!\!\!\!\!\!\!\!\!\!\!\!\!\!\!\!\!\!\!\!\!\!
	   	c_j(x^{\bar{\kappa}}(k)) \geq 0, \;\; \forall j \!\in\! \{1, \dots, n_c\},       \forall (k,\bar{\kappa})\!\in\! \mathbb{P}_j. \label{eq:fun_cons}
	\end{align} 
\end{subequations} 
 We next summarize the result w.r.t. Problem \ref{prob}.
\begin{mythm}\label{thm:1}\upshape
	The solution $u^*(0:T-1)$ of the optimization problem in \eqref{eq:opt} solves Problem \ref{prob} if and only if a solution to Problem \ref{prob} exists.
\end{mythm}
\begin{proof}
    Note that the main difference between Problem \ref{prob} and the optimization problem in \eqref{eq:opt} is in constraint \eqref{eq:fun_cons}. By the observation in Proposition \ref{prop:1} and the definition of the instant-shift pair sets, we know that $c_j(x^{\bar{\kappa}}(k)) \geq 0$  if and only if  $c_j(x^{\bar{\kappa}'}(k')) \geq 0$ for all $(k',\bar{\kappa}')\in I(k,\bar{\kappa})$ so that  the result follows by construction.
\end{proof}

\textbf{Computational complexity of \eqref{eq:opt}.}  The optimization problem in \eqref{eq:opt}  is convex if the input domain $\mathcal{U}$ is convex and if the system dynamics $f$, the cost $J$, and all the negated constraint $-c_j$ are convex functions. In general, however, the optimization problem can be non-convex, e.g., when $c_j$ encodes collision avoidance constraints in which case the reduction of constraints in \eqref{eq:fun_cons} is beneficial. In fact, the main complexity of the optimization problem \eqref{eq:opt} largely depends on the number of constraints in \eqref{eq:fun_cons}. Let us analyze the worst case when $\mathbb{T}_j=[0,T]$ in which case there are 
\begin{align}\label{eq:complexity}
    \Sigma_{j=1}^{n_c}|\mathbb{P}_j| = n_c\Big( (T+1)\Theta^N-T(\Theta-1)^N\Big)
\end{align}
constraints which is strictly less than $n_c (T+1)\Theta^N$ as in \eqref{eq:cons}. In fact, our encoding uses $n_c T (\Theta-1)^N$ less constraints. In practice, note that usually $T$ is much larger than the length of the ATR bound $\Theta$ and the number of agents $N$.
The result in equation \eqref{eq:complexity} can be derived as follows. 
The cardinality of the set $\mathbb{P}_j$ is  equivalent to the cardinality of set $\mathbb{P}_j^s$.
Assume that we construct $I(k,\bar{\kappa})$ in $\mathbb{P}_j^s$ recursively starting from $k=0$ (until $T$). At instant $k=0$, there will be $\Theta^N$ sets since there is no redundancy, i.e., for all $\bar{\kappa} \in \bar{\mathbf{K}}$ no two pairs $(0, \bar{\kappa})$ will be in the same instant-shift pair set. 
 At instant $k, k\!\in\! [1,T]$, $\forall \kappa_1, \dots, \kappa_N \in [\theta_1, \theta_2-1]$, the pair $(k, \bar{\kappa})$ is contained already within the sets constructed in the previous instants, i.e., only $\Theta^N-(\Theta-1)^N$ new sets have to  be constructed. We can repeat this argument and find that there are in total $\Theta^N + T(\Theta^N-(\Theta-1)^N)$ instant-shift pair sets in $\mathbb{P}_j^s$ (instant-shift pairs in $\mathbb{P}_j$) as in equation \eqref{eq:complexity}.

%%%%%%%%%%%%%%%%%%%%%%%%%%%%%%%%%%%%%%%%%%%%%%%%%%%%%%%%%%%%%%%
\section{Temporally Robust Control for  STL Specifications}\label{sec:sol_stl}

Solving Problem \ref{prob:stl} is inherently more challenging than solving Problem \ref{prob} due to the complexity that is added by considering STL tasks. We will rely on the Mixed Integer Linear Program (MILP) formulation proposed and used in \cite{raman2014model,raman2015reactive, sadraddini2018formal, yu2022model} to encode STL tasks\footnote{It is worth mentioning that other sound and complete MILP encoding frameworks for STL, e.g., \cite{kurtz2022mixed}, can similarly be used.}, and again incorporate the idea of instant-shift pair sets to leverage redundancy in time shifts  to enforce the ATR bound $[\theta_1,\theta_2]$. 

\textbf{Temporally robust STL encoding in an MIP.} Let us present the final result upfront. 
Problem \ref{prob:stl} can be solved by  the following optimization problem 
\begin{subequations}\label{eq:opt-stl}
	\begin{align}
		& \underset{u^*(0:T-1)}{\text{minimize}} & & J(x(0), u(0\!:\!T-1)) \\
		& \text{subject to} & & \nonumber \\
		& &&
		\!\!\!\!\!\!\!\!\!\!\!\!\!\!\!\!\!\!\!\!\!\!\!
		u(0),u(1),\dots, u(T-1) \in \mathcal{U}, \\
		& &&
		\!\!\!\!\!\!\!\!\!\!\!\!\!\!\!\!\!\!\!\!\!\!\!
		x(k+1) = f(x(k), u(k)), k=0, \dots, T-1, \\
		& &&
		\!\!\!\!\!\!\!\!\!\!\!\!\!\!\!\!\!\!\!\!\!\!\!
		z_{I(0, \bar{\kappa})}^{\phi} = 1, \ \forall \bar{\kappa} \in \bar{\mathbf{K}} \label{fundd},
	\end{align} 
\end{subequations} 
where  $z_{I(0, \bar{\kappa})}^{\phi}$ in equation \eqref{fundd}  are binary variables (formally presented in the remainder) that encode whether or not the formula $\phi$ is satisfied by the signal $x^{\bar{\kappa}}$. We remark that $z_{I(0, \bar{\kappa})}^{\phi}=1$ will correspond to satisfaction of $\phi$, i.e., to $\beta^\phi(x^{\bar{\kappa}})=1$. We further note that setting $k=0$ in $I(0, \bar{\kappa})$ indicates that we are imposing the specification $\phi$ at time zero. By considering all time shifts $\bar{\kappa} \in \bar{\mathbf{K}}$ in equation \eqref{fundd}, we then aim to robustly enforce satisfaction of $\phi$ for all trajectories shifted by $\bar{\kappa}$.
The encoding of the binary variables $z_{I(0, \bar{\kappa})}^{\phi}$ is introduced next.

 The general idea in \cite{raman2014model} is to recursively encode the STL formula $\phi$ starting from predicates $\pi^\mu$, and then encode the remaining Boolean and temporal operators. We additionally need to capture the time shift  $x^{\bar{\kappa}}$ in this encoding, which we next describe step-by-step.
In the case of STL tasks, we denote by $\mathbb{P}^s$ the
set of all instant-shift pair sets in terms of $[0,T_\phi]$, i.e.,  $\mathbb{P}^{s} := \{I(k,\bar{\kappa}) \mid  (k,\bar{\kappa}) \!\in\! [0,T_\phi]\times \bar{\mathbf{K}} \}$.

 \emph{Predicates:}
	For each predicate $\pi^\mu$ and for each instant-shift pair set $I(k, \bar{\kappa}) \in \mathbb{P}^s$, we introduce binary variables $z_{I(k, \bar{\kappa})}^{\mu} \!\in\! \{0,1\}$. 
    The following constraints enforce that $z_{I(k, \bar{\kappa})}^{\mu} = 1$ if and only if $\mu(x^{\bar{\kappa}}(k)) \geq 0$:
	\begin{align}
		\mu(x^{\bar{\kappa}}(k)) & \leq M z_{I(k, \bar{\kappa})}^{\mu} - \epsilon, \nonumber \\
		- \mu(x^{\bar{\kappa}}(k)) & \leq M (1-z_{I(k, \bar{\kappa})}^{\mu}) -\epsilon, \nonumber 
	\end{align}
	where $(k,\bar{\kappa})$ in $x^{\bar{\kappa}}(k)$ can  again be any element from the set $I(k,\bar{\kappa})$, and $M$ and $\epsilon$ are sufficiently large and small positive constants, respectively, see \cite{bemporad1999control} for details. We note that
	$z_{I(k, \bar{\kappa})}^{\mu} = 1$ implies that the predicate $\pi^\mu$ is satisfied at instant $k'$ after time shift $\bar{\kappa}'$ for all $(k', \bar{\kappa}') \!\in\! I(k, \bar{\kappa})$ due to Proposition \ref{prop:1} and the definition of $I(k, \bar{\kappa})$.  Otherwise, i.e., $z_{I(k, \bar{\kappa})}^{\mu} = 0$, the predicate $\pi^\mu$ is violated for the corresponding pairs. We remark that the use of instant-shift pair sets reduces the amount of binary variables significantly compared with considering all $(k,\bar{\kappa}) \!\in\! [0, T_\phi] \times \bar{\mathbf{K}}$. We discuss the overall computational complexity at the end of this section.

	\begin{remark}
		With some abuse of notation, we may write $z_{I(k', \bar{\kappa}')}^{\mu}$ instead of $z_{I(k, \bar{\kappa})}^{\mu}$ in the remainder if $(k', \bar{\kappa}') \!\in\! I(k, \bar{\kappa})$ since they both encode the same predicate.
		In other words, $z_{I(k', \bar{\kappa}')}^{\mu}$ and $z_{I(k, \bar{\kappa})}^{\mu}$ point to the same variable.
	\end{remark}

	\emph{Boolean operators:}
	% As described above, the binary variable $z_{I(k, \bar{\kappa})}^{\mu}$ of predicate $\mu$ equals to 1 if $\mu$ is true at instant $k$ after shift $\bar{\kappa}$ for all pairs in set $I(k, \bar{\kappa})$.
	Now consider that the formula $\phi$ consists of Boolean operators over formulae $\phi_i$ that are associated with variables $z_{I(k, \bar{\kappa})}^{\phi_i}$ whose value is 1 if and only if $\phi_i$ is satisfied by $x$ at time $k$ after time shift $\bar{\kappa}$. Equivalently, we can say that $\phi_i$ is satisfied by $x$ at time $k'$ after time shift $\bar{\kappa}'$ for all pairs $(k', \bar{\kappa}') \!\in\! I(k, \bar{\kappa})$. We introduce binary variables $z_{I(k, \bar{\kappa})}^{\phi} \!\in\! \{0,1\}$ of $\phi$ for a instant-shift pair set $I(k, \bar{\kappa}) \in \mathbb{P}^s$, and define the following set of constraints for different Boolean operators: 
	\begin{itemize}
		\item Negation $\phi=\neg \phi_1$: $z_{I(k, \bar{\kappa})}^{\phi} = 1-z_{I(k, \bar{\kappa})}^{\phi_1}$,
		\item Conjunction $\phi=\wedge_{i=1}^m \phi_i$: 
		\begin{align}
			& z_{I(k, \bar{\kappa})}^{\phi} \leq z_{I(k, \bar{\kappa})}^{\phi_i}, i=1,\dots, m, \nonumber \\
			& z_{I(k, \bar{\kappa})}^{\phi} \geq 1-m+\Sigma_{i=1}^m z_{I(k, \bar{\kappa})}^{\phi_i}, \nonumber
		\end{align}
		\item Disjunction $\phi=\vee_{i=1}^m \phi_i$: 
		\begin{align}
			& z_{I(k, \bar{\kappa})}^{\phi} \geq z_{I(k, \bar{\kappa})}^{\phi_i}, i=1,\dots, m, \nonumber \\
			& z_{I(k, \bar{\kappa})}^{\phi} \leq \Sigma_{i=1}^m z_{I(k, \bar{\kappa})}^{\phi_i}. \nonumber
		\end{align}
	\end{itemize}
	These constraints again enforce that $z_{I(k, \bar{\kappa})}^{\phi} =1$ if and only if $\phi$ is satisfied at time $k$ after time shift $\bar{\kappa}$.
 
	 \emph{Temporal operators:} The encoding of temporal operators follows the same idea by introducing binary variables $z_{I(k, \bar{\kappa})}^{\phi} \!\in\! \{0,1\}$ for a instant-shift pair set $I(k, \bar{\kappa}) \in \mathbb{P}^s$ as:
	\begin{itemize}
		\item Always $\phi=\mathbf{G}_{[a,b]}\phi_1$:
		\begin{align}
			z_{I(k, \bar{\kappa})}^{\phi} = \bigwedge_{k'=k+a}^{k+b} z_{I(k', \bar{\kappa})}^{\phi_1}, \nonumber
		\end{align}
		\item Eventually $\phi=\mathbf{F}_{[a,b]}\phi_1$:
		\begin{align}
			z_{I(k, \bar{\kappa})}^{\phi} = \bigvee_{k'=k+a}^{k+b} z_{I(k', \bar{\kappa})}^{\phi_1}, \nonumber
		\end{align}
		\item Until $\phi=\phi_1\mathbf{U}_{[a,b]}\phi_2$:
		\begin{align}
			z_{I(k, \bar{\kappa})}^{\phi} = \bigvee_{k'=k+a}^{k+b} (z_{I(k', \bar{\kappa})}^{\phi_2} \wedge \bigwedge_{k''=k}^{k'} z_{I(k'', \bar{\kappa})}^{\phi_1}).\nonumber
		\end{align}
	\end{itemize}

\textbf{Soundness and completeness of \eqref{eq:opt-stl}.} By the recursive construction of $z_{I(0, \bar{\kappa})}^{\phi}$ via an MIP encoding and by induction over the structure of the STL formula $\phi$, we know that $\phi$ is satisfied by $x^{\bar{\kappa}}$ if and only if we enforce that $z_{I(0, \bar{\kappa})}^{\phi}=1$. Note again that we set $k=0$ to impose that the STL specification $\phi$ is satisfied at time $k=0$ \cite{raman2014model}. 
% By the observation in Proposition \ref{prop:1} and the definition of the instant-shift pair sets, we know again that  $z_{(0, \bar{\kappa})}^{\phi}=1$ implies that $z_{(0, \bar{\kappa}')}^{\phi}=1$ for all $(0,\bar{\kappa}')\in I(0,\bar{\kappa})$.
By enforcing $z_{I(0, \bar{\kappa})}^{\phi}=1$ for all $\bar{\kappa} \!\in\! \bar{\mathbf{K}}$ as in equation \eqref{fundd}, we indeed achieve the temporally robust satisfaction of $\phi$.
% \begin{equation}\label{eq:z_0^phi}
% 	z_{I(k, \bar{\kappa})}^{\phi} = 1, \ \forall I(k, \bar{\kappa}) \in \mathbb{P}_{\{0\}, \bar{\mathbf{K}}}.
% \end{equation}
%  Problem \ref{prob:stl} can be solved by  the following optimization problem 
% \begin{subequations}\label{eq:opt-stl}
% 	\begin{align}
% 		& \underset{u^*(0:T-1)}{\text{minimize}} & & J(x(0), u(0\!:\!T-1)) \\
% 		& \text{subject to} & & \nonumber \\
% 		& &&
% 		\!\!\!\!\!\!\!\!\!\!\!\!\!\!\!\!\!\!\!\!\!\!\!
% 		u(0),u(1),\dots, u(T-1) \in \mathcal{U}, \\
% 		& &&
% 		\!\!\!\!\!\!\!\!\!\!\!\!\!\!\!\!\!\!\!\!\!\!\!
% 		x(k+1) = f(x(k), u(k)), k=0, \dots, T-1, \\
% 		& &&
% 		\!\!\!\!\!\!\!\!\!\!\!\!\!\!\!\!\!\!\!\!\!\!\!
% 		z_{I(0, \bar{\kappa})}^{\phi} = 1, \ \forall \bar{\kappa} \in \bar{\mathbf{K}} \label{fundd}.
% 	\end{align} 
% \end{subequations} 

We next summarize the result w.r.t. Problem \ref{prob:stl}.

\begin{mythm}\upshape
The solution $u^*(0:T-1)$ of the optimization problem in \eqref{eq:opt-stl} solves Problem \ref{prob:stl} if and only if Problem \ref{prob:stl} has a solution.
\end{mythm}

\begin{proof}
Similar to Theorem \ref{thm:1}, the main difference between Problem \ref{prob:stl} and the optimization problem  \eqref{eq:opt-stl} is in constraint \eqref{fundd}. By the observation in Proposition \ref{prop:1} and the definition of the instant-shift pair sets, it follows that $\beta^\phi(x^{\bar{\kappa}})=1$ for each $\bar{\kappa}\in [\theta_1,\theta_2]^N$. The MILP encoding from \cite{raman2014model} is known to be exact for appropriately selected $M$ and $\epsilon$ (as we assume),  so that our encoding for the ATR bound $[\theta_1,\theta_2]$ is necessary and sufficient as we consider all possible time shifts by the instant-shift pair set.
\end{proof}
%  \begin{figure*}[htbp]
% 	\centering
	
% 	\subfloat{
%         \rotatebox{90}{\scriptsize{~~~~~~~~Function Constraints Task}}
% 		\begin{minipage}[t]{0.23\linewidth}
% 			\centering
% 			\includegraphics[width=1\linewidth]{test.pdf}
% 		\end{minipage}
% 	}
% 	\subfloat{
% 		\begin{minipage}[t]{0.23\linewidth}
% 			\centering
% 			\includegraphics[width=1\linewidth]{test.pdf}
% 		\end{minipage}
% 	}
% 	\subfloat{
% 		\begin{minipage}[t]{0.23\linewidth}
% 			\centering
% 			\includegraphics[width=1\linewidth]{test.pdf}
% 		\end{minipage}
% 		\begin{minipage}[t]{0.23\linewidth}
% 			\centering
% 			\includegraphics[width=1\linewidth]{test.pdf}
% 		\end{minipage}
% 	} \\
	
% 	\vspace{-3pt}
% 	\setcounter{subfigure}{0}

%     \subfloat[$\theta_1=0, \theta_2=0$]{
% 		\rotatebox{90}{\scriptsize{~~~~~~~~~~~~~~~~~STL Task}}
% 		\begin{minipage}[t]{0.23\linewidth}
% 			\centering
% 			\includegraphics[width=1\linewidth]{test.pdf}
% 		\end{minipage}
% 	}
% 	\subfloat[$\theta_1=-1, \theta_2=1$]{
% 		\begin{minipage}[t]{0.23\linewidth}
% 			\centering
% 			\includegraphics[width=1\linewidth]{test.pdf}
% 		\end{minipage}
% 	}
% 	\subfloat[$\theta_1=-2, \theta_2=2$]{
% 		\begin{minipage}[t]{0.23\linewidth}
% 			\centering
% 			\includegraphics[width=1\linewidth]{test.pdf}
% 		\end{minipage}
% 	}
%         \subfloat[$\theta_1=-3, \theta_2=3$]{
% 		\begin{minipage}[t]{0.23\linewidth}
% 			\centering
% 			\includegraphics[width=1\linewidth]{test.pdf}
% 		\end{minipage}
% 	}
% 	\caption{Simulation results of the motion planning case in Section \ref{subsec:case2}.}
% 	\label{fig:sim}
% \end{figure*}

\begin{figure*}[htbp]
	\centering
	
	\subfloat{
        \rotatebox{90}{\scriptsize{ Constraint Function Task}}
		\begin{minipage}[t]{0.23\linewidth}
			\centering
			\includegraphics[width=1\linewidth]{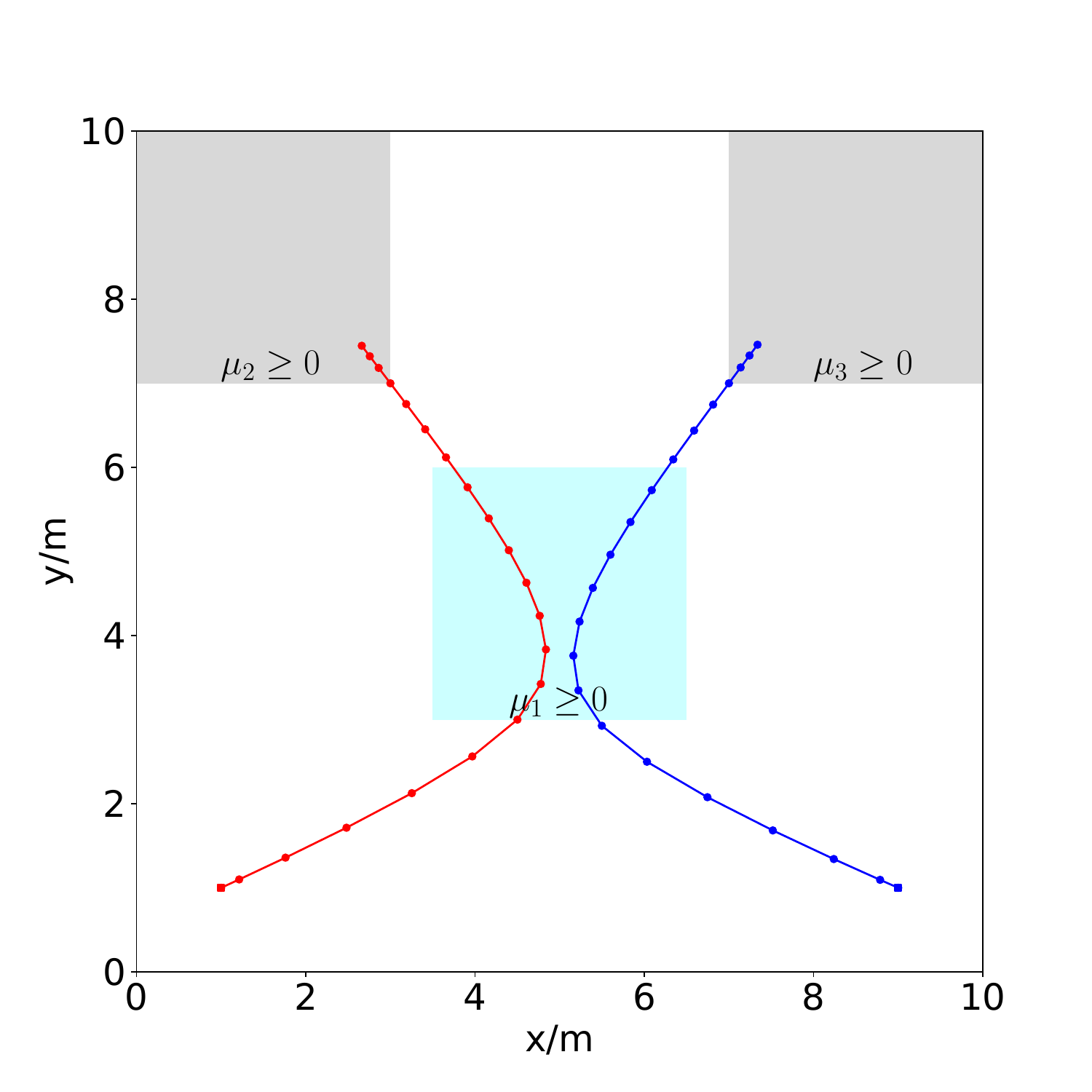}
		\end{minipage}
	}
	\subfloat{
		\begin{minipage}[t]{0.23\linewidth}
			\centering
			\includegraphics[width=1\linewidth]{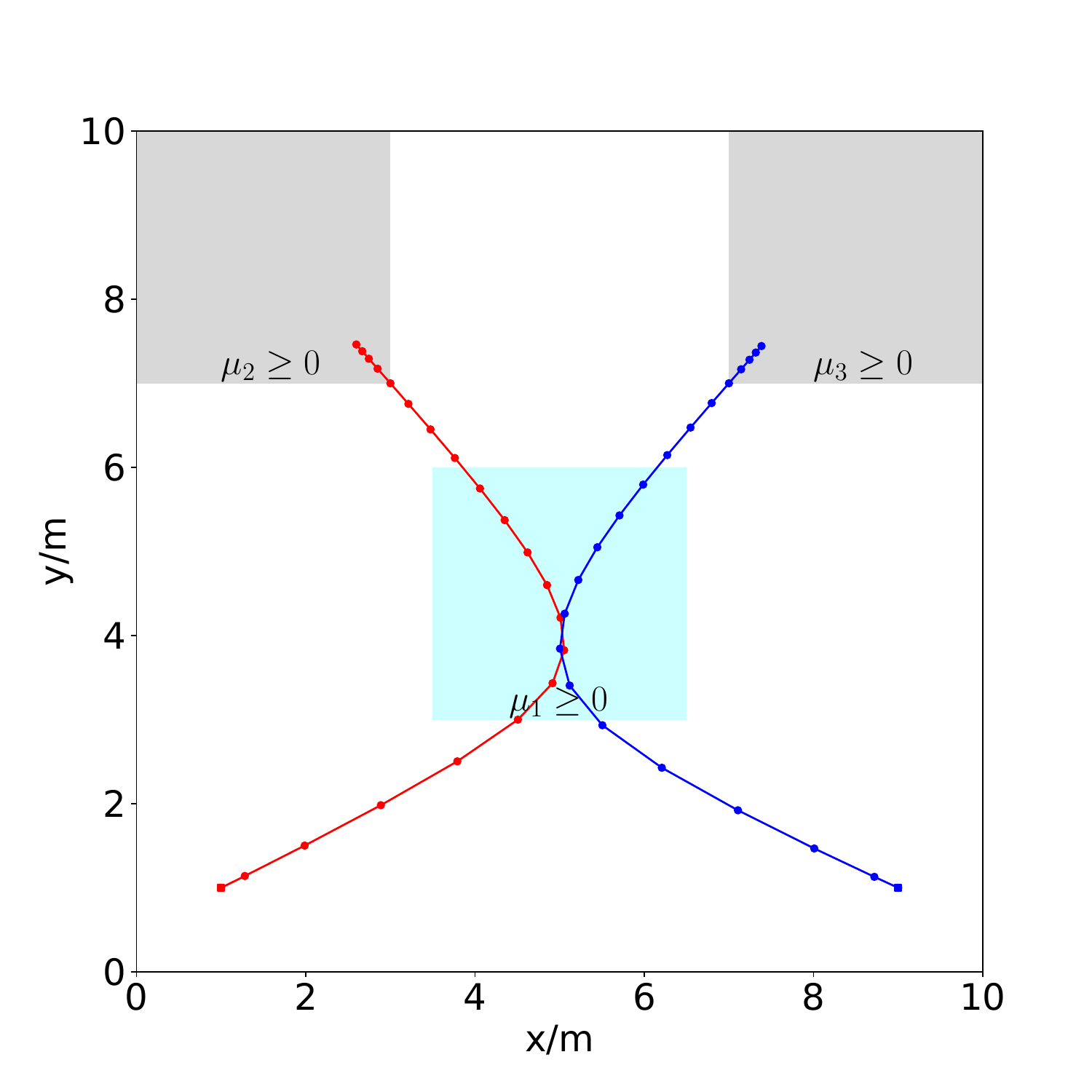}
		\end{minipage}
	}
	\subfloat{
		\begin{minipage}[t]{0.23\linewidth}
			\centering
			\includegraphics[width=1\linewidth]{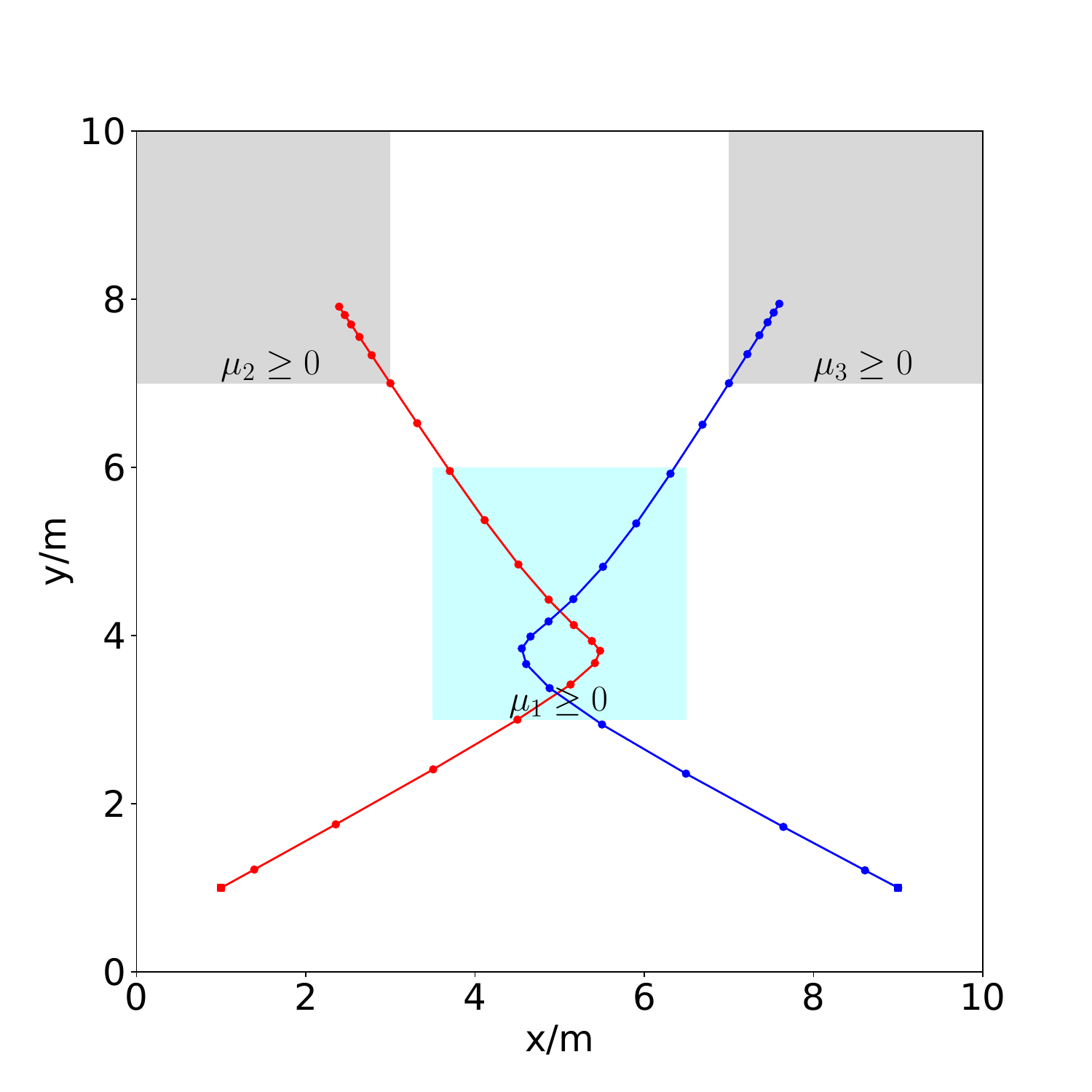}
		\end{minipage}
		\begin{minipage}[t]{0.23\linewidth}
			\centering
			\includegraphics[width=1\linewidth]{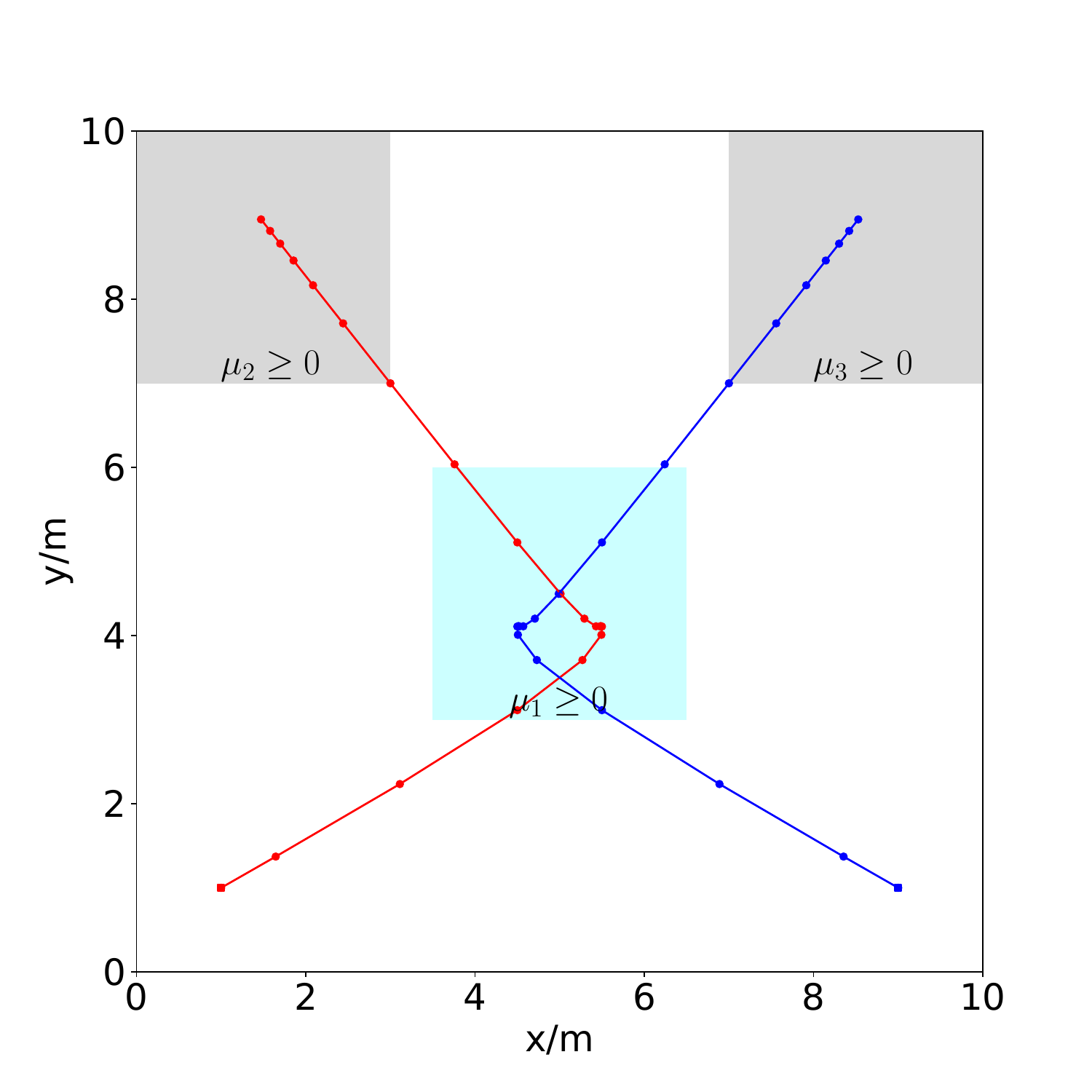}
		\end{minipage}
	} \\
	
	\vspace{-3pt}
	\setcounter{subfigure}{0}

    \subfloat[$\theta_1=0, \theta_2=0$]{
		\rotatebox{90}{\scriptsize{~~~~~~~~~~~~STL Task}}
		\begin{minipage}[t]{0.23\linewidth}
			\centering
			\includegraphics[width=1\linewidth]{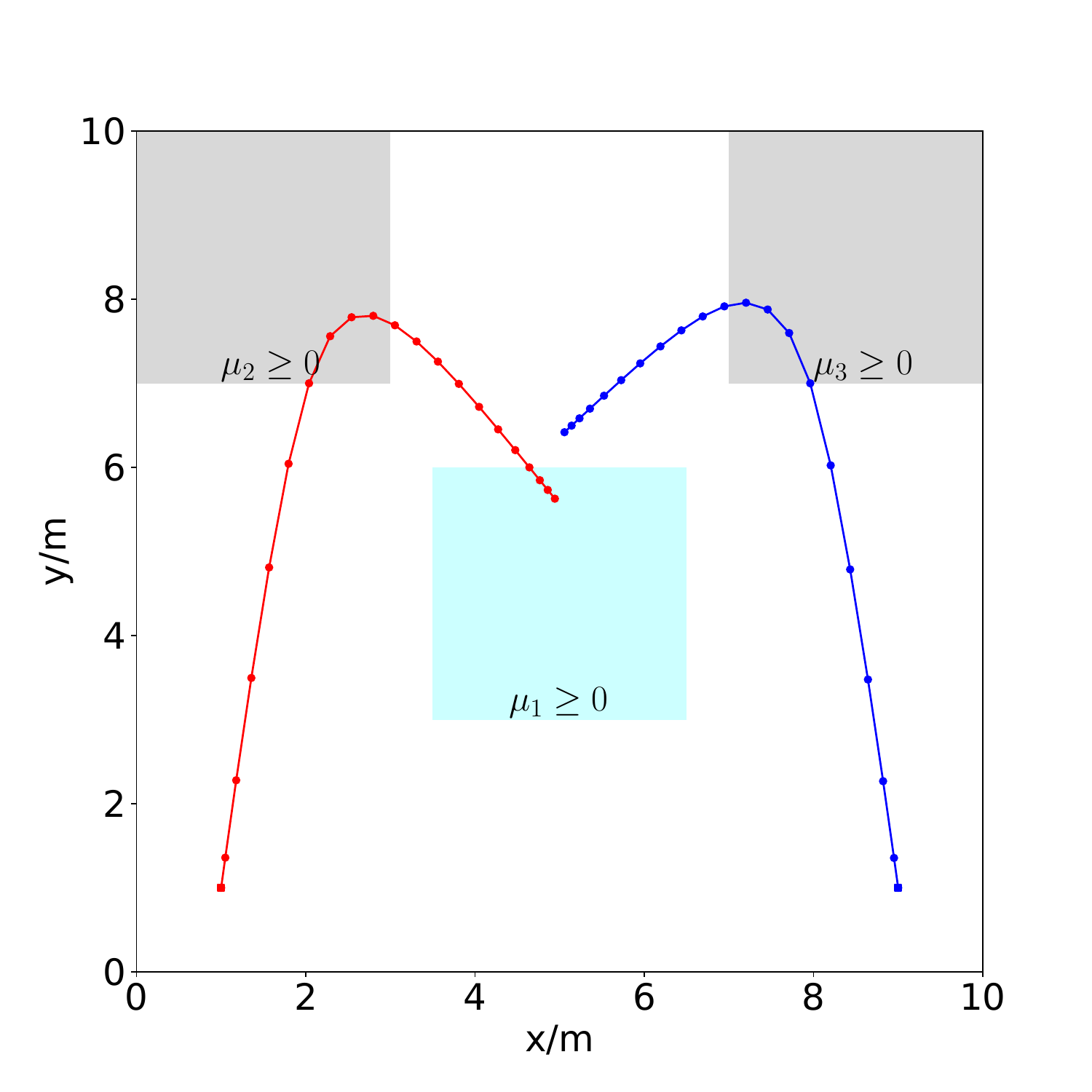}
		\end{minipage}
	}
	\subfloat[$\theta_1=-1, \theta_2=1$]{
		\begin{minipage}[t]{0.23\linewidth}
			\centering
			\includegraphics[width=1\linewidth]{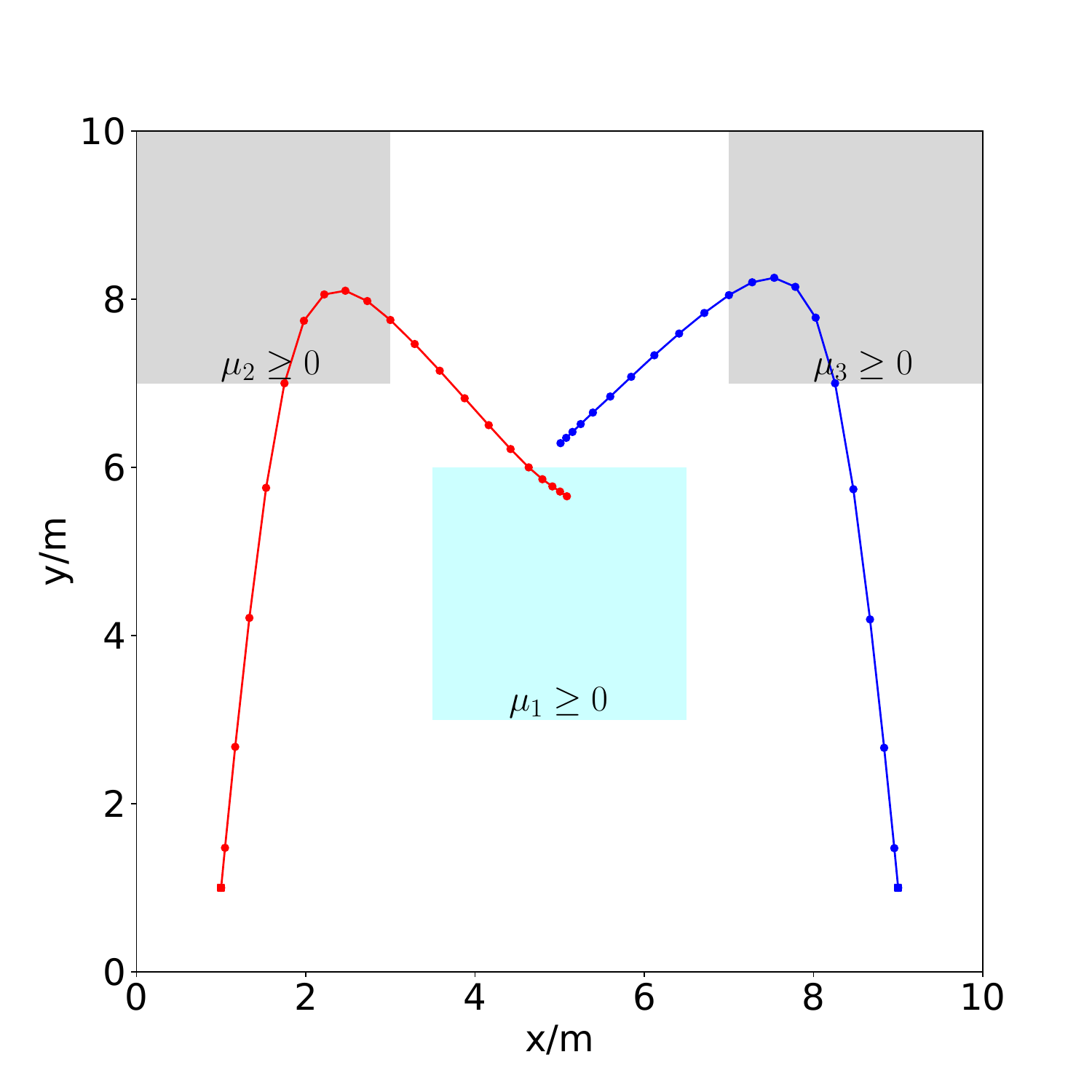}
		\end{minipage}
	}
	\subfloat[$\theta_1=-2, \theta_2=2$]{
		\begin{minipage}[t]{0.23\linewidth}
			\centering
			\includegraphics[width=1\linewidth]{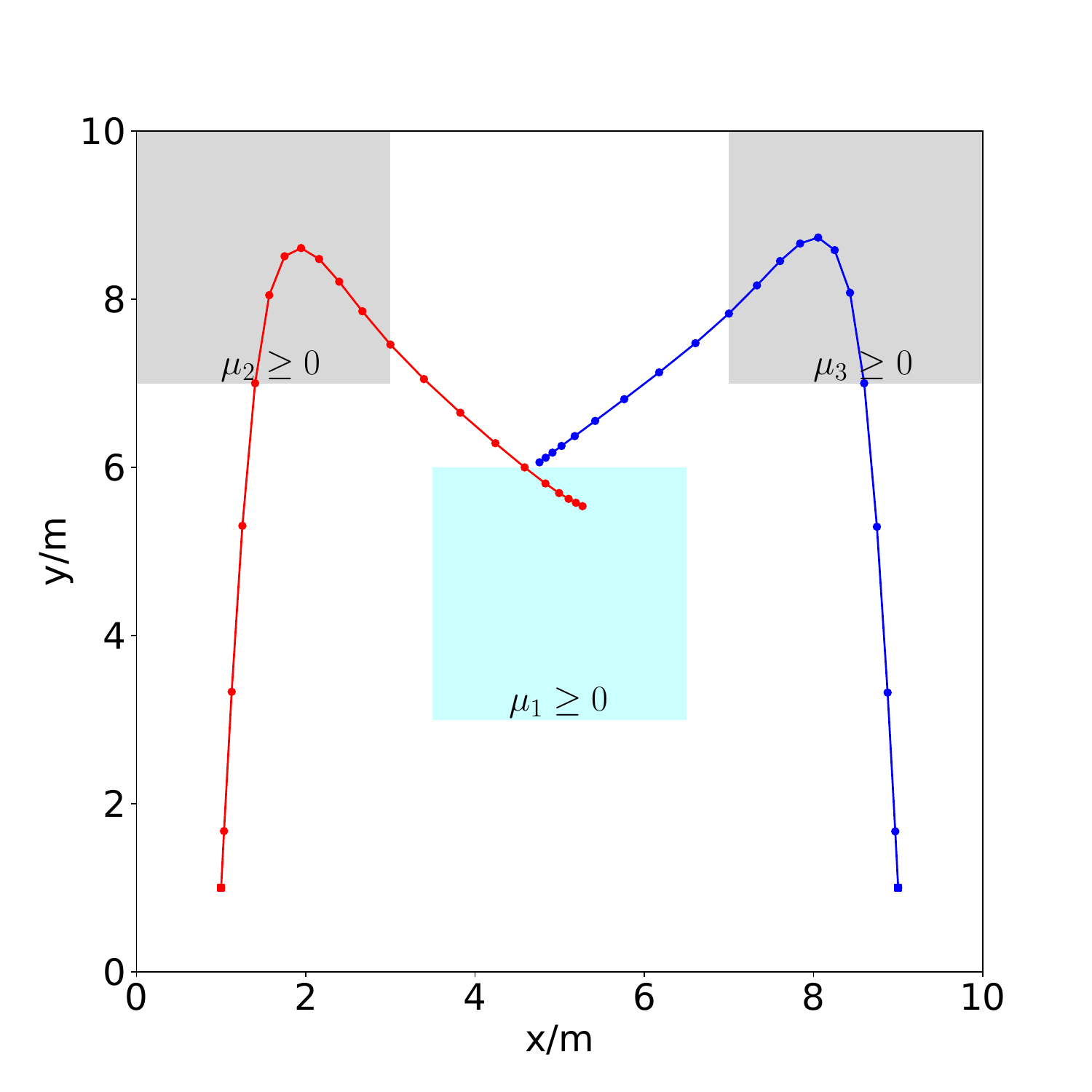}
		\end{minipage}
	}
        \subfloat[$\theta_1=-3, \theta_2=3$]{
		\begin{minipage}[t]{0.23\linewidth}
			\centering
			\includegraphics[width=1\linewidth]{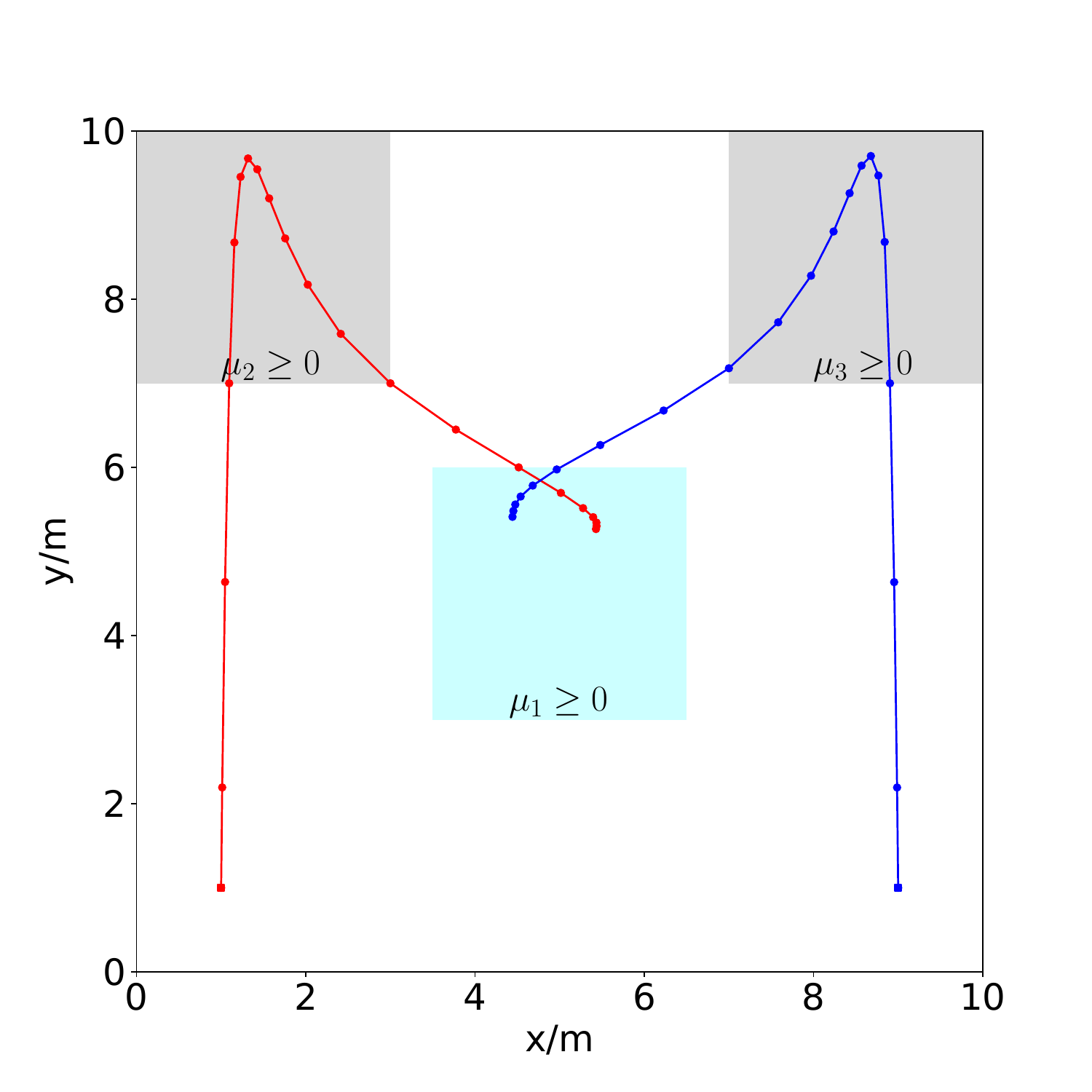}
		\end{minipage}
	}
	\caption{Simulation results of the motion planning case in Section \ref{subsec:case2}.}
	\label{fig:sim}
\end{figure*}

% \begin{figure*}[h]
% 	\centering
% 	\subfloat[Function task $\theta_1=0, \theta_2=0$]
%         {\includegraphics[height=4.3cm]{case1-(-1,1).pdf}}  
%         \subfloat[Function task $\theta_1=-1, \theta_2=1$]
%         {\includegraphics[height=4.3cm]{case1-(-1,1).pdf}} 
%         \subfloat[Function task $\theta_1=-2, \theta_2=2$]
%         {\includegraphics[height=4.3cm]{case1-(-2,2).pdf}} 
%         \subfloat[Function task $\theta_1=-3, \theta_2=3$]
%         {\includegraphics[height=4.3cm]{case1-(-3,3).pdf}} \\
% 	\subfloat[STL task $\theta_1=0, \theta_2=0$]
%         {\includegraphics[height=4.1cm]{case2.pdf}} 
%         \subfloat[STL task $\theta_1=0, \theta_2=0$]
%         {\includegraphics[height=4.1cm]{case2.pdf}}
%         \subfloat[STL task $\theta_1=0, \theta_2=0$]
%         {\includegraphics[height=4.1cm]{case2.pdf}}
%         \subfloat[STL task $\theta_1=0, \theta_2=0$]
%         {\includegraphics[height=4.1cm]{case2.pdf}}
% 	\caption{Simulation results of the motion planning case in Section \ref{subsec:case2}.}
% 	\label{fig:sim}
% \end{figure*}

\textbf{Computational complexity of \eqref{eq:opt-stl}.}
	The complexity of the optimization problem in \eqref{eq:opt-stl} is in general that of a Mixed-Integer Programming (MIP) which is NP-hard. We can  use the number of binary variables to approximate the computational complexity since, in the case of convex cost function $J$ and dynamics $f$, the complexity mainly depends on the number of binary variables.
	% Note that, during the encoding process, all the variables for predicates $z_{I(k, \bar{\kappa})}^{\mu}$ are binary and the remaining variables for Boolean and temporal operations $z_{I(k, \bar{\kappa})}^{\psi}$ can be continuous.
    Here we analyze particularly the number of binary variables introduced for predicates. The total number of binary variables then recursively depends on the number of Boolean and temporal operators in the STL formula $\phi$, see e.g., \cite{raman2014model} for a discussion.
    In the worst case, similar to the case of  constraint functions $c$ in the last section, the above encoding  will introduce 
    \begin{align}
        n_{\pi}\Big((T_\phi+1)\Theta^N - T_\phi(\Theta-1)^N\Big) \nonumber
    \end{align}
	% \begin{align}
	% 	n_{\pi}\Big(T_\phi\Theta^N - (T_\phi-1)(\Theta-1)^N\Big)	\nonumber
	% \end{align}
	binary variables for predicates based on the instant-shift pair set $I(k,\bar{\kappa})$, which is much less than the straightforward combinatorial method with $n_{\pi}(T_\phi+1) \Theta^N$ binary variables, where $n_{\pi}$ is the number of predicates in the formula. 
    % \YXY{it may be $n_{\pi}\Big((T_\phi+1)\Theta^N - T_\phi(\Theta-1)^N\Big)$ instead of $n_{\pi}\Big(T_\phi\Theta^N - (T_\phi-1)(\Theta-1)^N\Big)$, since the time period of predicate we considered is $[0, T_\phi]$.}

\section{Case Study}\label{sec:case}

In this section, we illustrate our efficient method on two illustrative case studies. 
All simulations are conducted in \textsf{Python 3}  and we use \textsf{Gurobi} \cite{optimization2018gurobi} to solve the optimization problem. 
Simulations are carried out on a laptop computer with i7-8565 CPU and 8 GB of RAM.
Our implementations are available at \url{https://github.com/Xinyi-Yu/ATR-cases}, where more details can be found.

\subsection{Computation Time Analysis}\label{subsec:case1}
First, we present a set of case studies to analyze the computation time and how many constraints and variables are introduced in the corresponding optimization problems. For this purpose, we consider the following function constraints and STL formulae, 
\begin{align}
    & c_1(x,k) := \left\{
    \begin{array}{cl}
        x_1 + x_2 -1, & \text{for  all } k \in [5,10] \\
        x_1 - x_2 -2, & \text{for  all } k \in [17,25] \\
        1, & \text{otherwise}
    \end{array}
    \right. \nonumber \\
    & c_2(x,k) := \left\{
    \begin{array}{cl}
        (x_1 - x_2)x_3 - 1, & \text{for  all } k \in [5,6] \\
        x_1 + x_2 + x_3 - 3, & \text{for  all } k \in [11,15] \\
        1, & \text{otherwise}
    \end{array}
    \right. \nonumber \\
    & \phi_1 := \mathbf{G}_{[1, 10]}(x_1-x_2)^2<1, \nonumber \\
    & \phi_2 :=  \mathbf{G}_{[7, 15]}(x_1-x_2+x_3)>0.5 \wedge \mathbf{F}_{[1,9]} x_1(x_2-x_3)<0.2, \nonumber \\
    & \phi_3 := \mathbf{F}_{[3, 15]} \mathbf{G}_{[0,8]} (x_1 + x_2 + x_3) >1. \nonumber
\end{align}

We consider simple  agent dynamics $x_i(k+1) = x_i(k)+u_i(k)$ with initial state $x(0) :=0$ and where $x_i$ and $u_i$ are one-dimensional states and inputs with  constraint sets $\mathcal{X}_i:=[-2,2]$ and $\mathcal{U}_i:=[-2,2]$. 
For the tasks of $c_1(x,k)$ and $\phi_1$, we consider $N=2$ agents, and for the remaining cases we consider $N=3$ agents. Furthermore, the required bounds for $c_1$, $c_2$, $\phi_1$, $\phi_2$ and $\phi_3$ are $[-4,4]$, $[-1,1]$, $[-3,3]$, $[-1,2]$ and $[-3,4]$ respectively. 
% for the cases of $c_1(x,k)$, $\phi_1$ and $\phi_2$, we require the ATR bound to be $[-1,1]$, and $[-2,1]$ for the other two cases.
The cost function is $J := \Sigma_{k=0}^{T-1} u(k)^T  u(k)$ where $T$ is the final time.

In Table \ref{table:1}, we report the solve times for all five case studies as well as the number of constraints and binary variables that we introduced. We also compare our results to a naive approach where we combinatorially explore all time shifts as discussed in the beginning of Section \ref{sec:sol_c}.
As shown in Table \ref{table:1}, our proposed method is more efficient and performs better than the naive approach, especially when it comes to more complex tasks (e.g., $c_2$, $\phi_2$ and $\phi_3$).

\begin{table}[htb] 
    \begin{center}    
    \caption{Results report of academic cases in Section \ref{subsec:case1}}
    \label{table:1}  
    \begin{threeparttable}[b]
        \begin{tabular}{cccccccc}
            \toprule
            \multirow{2}{*}{\textbf{Tasks}} & \multicolumn{2}{c}{\textbf{Naive method}} & \multicolumn{2}{c}{\textbf{Our method}} \\
            \cmidrule(lr){2-3} \cmidrule(lr){4-5} 
            & Number\tnote{1} &  Solver time (s)  & Number\tnote{1} & Solver time (s) \\
            \midrule
            $c_1$ & 1215 & 0.0404& 383  & 0.0169  \\
            $c_2$ & 189 & 11.2559 & 149 & 1.0292  \\
			$\phi_1$ & 539 & 0.0424 & 215  & 0.0109   \\
			$\phi_2$ & 1280 & 15.858 & 875 & 1.5807   \\
			$\phi_3$ & 17920 & 84.001 & 6944 & 6.7040   \\
            \bottomrule
        \end{tabular}
        \begin{tablenotes}
        \item[1] Number means the number of inequalities for function constraint $c_1$ and $c_2$ while it means the number of binary variables for STL formulae $\phi_1$, $\phi_2$ and $\phi_3$.
        \end{tablenotes}
    \end{threeparttable}
    \end{center}   
\end{table}

\subsection{Multi-Robot Case Study}\label{subsec:case2}

Consider again the workspace in the motivating example in Section \ref{sec:intro}. Recall that we consider two agents, and assume that both agents have the same dynamics which we model as a double integrator system with a sampling period of $1$ second, i.e., that
\[
    x(k+1) = 
	\begin{bmatrix}
		1 & 1 & 0 & 0 \\
		0 & 1 &  0 & 0 \\
		0 & 0 & 1 & 1 \\
		0 & 0 & 0 & 1 
	\end{bmatrix} x(k) + 
	\begin{bmatrix} 
		0.5 & 0   \\ 
		1 & 0 \\ 
		0 & 0.5 \\ 
		0 & 1 
	\end{bmatrix} u(k),
\]
where $x = [p_x \ v_x \ p_y \ v_y]^T$ denotes $x$-position, $x$-velocity, $y$-position and $y$-velocity, and control input $u = [u_x \ u_y]^T$ denotes $x$-acceleration and $y$-acceleration, respectively. The state and input constraints are $\mathcal{X} = [0,10] \times [-2.5, 2.5] \times [0,10] \times [-2.5, 2.5]$ and $\mathcal{U} = [-2.5, 2.5]^2$.
We stack both models together so that we have $n=8$ and  $N=2$.

\textbf{Tasks:} We consider two tasks represented by the constraint function and the STL formula from Examples 1 and 2, respectively.  The specific representations of $\mu_1(x), \mu_2(x), \mu_3(x)$ are rectangles shown as corresponding cyan and grey areas in Figure \ref{fig:sim}, and we set $D$ to be $1$ here.
The cost function is $J= J(x_1, u_1)+J(x_2, u_2)$ where $ J(x, u) = \Sigma_{k=0}^{20}0.1(v_x^2(k)+v_y^2(k))+0.9(u_x^2(k)+u_y^2(k))$.
% Furthermore, we require the asynchrony bound is $[\theta_1, \theta_2]=[-1,2]$.

\textbf{Results:} 
We consider the four different ATR bounds  $[\theta_1, \theta_2] = [0,0]$ (no ATR bound), $[-1,1]$, $[-2,2]$, and $[-3,3]$ respectively. The computed trajectories for both tasks are shown in Figure \ref{fig:sim}. 
The computation time for four ATR bounds in the case of Constraint functions are 0.0079s, 0.0129s, 0.0169s and 0.0210s while those of STL tasks are 0.2891s, 4.9307s, 40.3630s and 46.7640s.
First, for the constraint $c$ (top row), we can see that stricter ATR constraints (from left to right) lead to temporally more robust paths. For an ATR bound of $[-3,3]$, both agents will arrive earlier in the middle region and stay in it for a longer time with agents staying closer to each other. Second, for the STL task (bottom row), we can observe a similar behavior.

\section{Conclusion}\label{sec:con}
We proposed an efficient control synthesis method under asynchronous temporal robustness constraints  for high-level tasks described by  function constraints and signal temporal logic formulae. Given a temporal robustness bound, we compute a sequence of control inputs so that the specification is satisfied by the system as long as each sub-trajectory is shifted no more than the ATR bound. We avoid combinatorially exploring all shifted sub-trajectories by identifying redundancy between them. In simulations, we show the efficiency of our method and how we can synthesize temporally robust trajectories. In the future, we aim to explore to maximize the asynchronous temporal robustness and consider decentralized control synthesis for multi-agent systems. 

\bibliographystyle{ieeetr}
\bibliography{ATR}

\end{document}